\documentclass{article}
\usepackage[utf8]{inputenc}
\usepackage[margin=1in]{geometry}
\usepackage{amssymb,amsfonts,amsmath,amsthm,comment,mleftright,enumerate,graphicx,hyperref,xcolor}
\usepackage[makeroom]{cancel}

\newtheorem{theorem}{Theorem}
\newtheorem{lemma}{Lemma}
\newtheorem{corollary}{Corollary}

\newcommand{\paren}[1]{\left(#1\right)}
\newcommand{\bracket}[1]{\left[#1\right]}
\renewcommand{\brace}[1]{\left\{#1\right\}}

\newcommand{\Z}{\mathbb{Z}}
\newcommand{\ZM}[1]{\Z/#1\Z}
\newcommand{\F}{\mathbb{F}}

\newcommand{\M}[1]{\begin{bmatrix}#1\end{bmatrix}}
\newcommand{\MA}[2]{\mleft[\begin{array}{#1}#2\end{array}\mright]}
\newcommand{\casew}[1]{\begin{cases}#1&\textrm{else}\end{cases}}

\newcommand{\rank}[1]{\mathrm{rank}\paren{#1}}
\newcommand{\rowspan}[1]{\mathrm{rowspan}\paren{#1}}
\newcommand{\rref}[1]{\mathrm{rref}\paren{#1}}
\newcommand{\GL}[2]{\mathrm{GL}\paren{#1,\ #2}}

\newcommand{\NAE}[1]{\mathrm{NAE}\paren{#1}}
\renewcommand{\~}[1]{\widetilde{#1}}
\renewcommand{\vec}[1]{\overrightarrow{#1}}
\renewcommand{\min}[1]{\mathrm{min}\paren{#1}}

\title{Low-Rank Tensor Decomposition over Finite Fields}
\author{Jason Yang}
\date{}

\begin{document}

\maketitle

\begin{abstract}
We show that finding rank-$R$ decompositions of a 3D tensor, for $R\le 4$, over a fixed finite field can be done in polynomial time.
However, if some cells in the tensor are allowed to have arbitrary values, then rank-2 is NP-hard over the integers modulo 2.
We also explore rank-1 decomposition of a 3D tensor and of a matrix where some cells are allowed to have arbitrary values.
\end{abstract}

\section{Introduction}
Given a $p\times q\times s$ tensor (multidimensional array) $T$, a rank-$R$ decomposition is a collection of variables $A_{a,b},B_{c,d},C_{e,f}$ such that

\[T_{i,j,k}=\sum_{r=0}^{R-1} A_{r,i}B_{r,j}C_{r,k} \ \forall i,j,k.\]

This equation is also written as $T=\sum_{r=0}^{R-1} A_{r,:}\times B_{r,:}\times C_{r,:}$, where $M_{r,:}$ denotes the $r$-th row of $M$ and $\times$ denotes the outer product: for arbitrary tensors $X$ and $Y$, $(X\times Y)_{i_0,\dots i_{k-1}, j_0,\dots j_{\ell-1}}
:=X_{i_0,\dots i_{k-1}}Y_{j_0,\dots j_{\ell-1}}$.

Tensor decomposition is at the heart of fast matrix multiplication, a problem that is the primary bottleneck of numerous linear algebra and graph combinatorics algorithms, such as matrix inversion and triangle detection. All asymptotically fast algorithms for matrix multiplication use a divide-and-conquer scheme, and finding an efficient scheme is equivalent to finding a decomposition of a certain tensor with low rank \cite{survey}.

While significant progress has been made on this problem over the decades, most of the asymptotically fastest algorithms discovered use indirect constructions and convergence arguments that result in enormously large divide-and-conquer schemes, making the constant factors for these algorithms far too big for them to be useful in practice.
For this reason, there has also been interest in finding explicit decompositions for the tensors of small divide-and-conquer schemes. With the introduction of AlphaTensor \cite{alphatensor} and flip graphs \cite{flip} just over a year ago, several new low-rank decompositions have been discovered, especially over finite fields of characteristic 2 \cite{adaflip}.

Decompositions using only integers are especially preferred, since they avoid having to use floating point arithmetic, and most decompositions of small divide-and-conquer schemes that have been found have happened to use only integers. This makes it useful to search for decompositions over finite fields (e.g., integers modulo 2), as doing so makes the search space finite and allows one to rule out large classes of integer solutions.

It is known that deciding whether a given tensor has a decomposition of rank $\le R$ is NP-hard over finite fields via a reduction from 3SAT \cite{hardness}. However, since we only want low-rank decompositions, we are more interested in hardness results for small fixed $R$.

\subsection{Results}
Our main result in this report is a polynomial-time algorithm for 3D tensor decomposition with rank $\le 4$.
For simplicity of asymptotic analysis, we focus on $n\times n\times n$ tensors $T$, although our algorithm easily generalizes to non-cube tensors.

\begin{theorem}
\label{main}
For fixed $R\le 4$, finding a rank-$R$ decomposition over a finite field $\F$ of an $n\times n\times n$ tensor, or determining that it does not exist, can be done in $O(n^3+f(|\F|,R)n^2)$ time, for some function $f$.
\end{theorem}

However, when allowing certain cells of the target tensor to have arbitrary value (``wildcards"), we have a hardness result over the integers mod 2 via a reduction from Not-All-Equal 3SAT (NAE-3SAT):
\begin{theorem}
\label{NAE}
If some cells of a $p\times q\times s$ tensor $T$ are allowed to have arbitrary value, then determining if $T$ has a rank-2 decomposition over $\ZM{2}$ is NP-hard.
\end{theorem}

We also sketch out polynomial-time algorithms for rank-1 decomposition with wildcards over a 3D tensor and a matrix, over an arbitrary fixed finite field.

\subsection{Notation}
\begin{itemize}
    \item We use NumPy notation for indexing on tensors.
    \begin{itemize}
        \item e.g. for a matrix $M$:
        \begin{itemize}
            \item $M_{i,:}$ denotes the $i$-th row
            \item $M_{:i,:}$ denotes the matrix truncated to its first $i$ rows
            \item $M_{:,i}$ denotes the $i$-th column
        \end{itemize}
    \end{itemize}
    \item All tensors and sequences are 0-indexed.
    \item The reduced row echelon form of a matrix $M$ is denoted as $\rref{M}$.
\end{itemize}

\section{Without wildcards}
\subsection{Preliminary results about matrix factorizations}
Before describing our algorithm for Theorem \ref{main}, we present several lemmas on \textit{matrix} rank decompositions that we will use extensively.

\noindent The first lemma is a direct result of Gaussian elimination:

\begin{lemma}
\label{rref}
Given any matrix $M\in\F^{m\times n}$, $\exists T\in\GL{m}{\F}$ s.t. $TM=\rref{M}$.
\end{lemma}
It follows that inverting a $n\times n$ matrix can be done in $O(n^3)$ time.

\noindent The next lemma is the uniqueness of rref (up to all-zero rows):
\begin{lemma}
\label{unique-rref}
Two rref matrices with the same rowspan must be equal, up to the presence of all-zero rows.
\end{lemma}
\begin{proof}
Consider an rref matrix $W=\MA{c}{w_0\\\hline \vdots \\\hline w_{k-1} \\\hline O}$, where each $w_i\ne\vec{0}$ has its leading 1 at column $\ell_i$, with $\ell_0<\dots<\ell_{k-1}$. Consider a linear combination of the rows, $v:=\sum_i \alpha_i w_i$, where the $\alpha_i$ are not all zero: then $v_{\ell_i}=\alpha_i$, since by definition of rref, every column of $W$ that has a leading 1 has all of its other elements set to 0.
As a consequence, if we are instead given some $v$ that is known to be a linear combination of $w_i$, then the list of coefficients $\alpha_i$ of that linear combination are unique and can immediately be deduced by reading each $v_{\ell_i}$ element.

Additionally, the leading nonzero term of $v$ must be at some $\ell_i$-th column, as we can use the following casework: if $\alpha_0\ne 0$, the leading nonzero term of $v$ is $\ell_0$, since each $w_i$ is all zeros to the left of its leading 1: else if $\alpha_1\ne 0$, the leading nonzero term of $v$ is $\ell_1$: etc.

Now suppose there is another rref matrix $W'$ s.t. $\rowspan{W}=\rowspan{W'}$: then $W'$ cannot have a leading 1 at a column that is not one of the $\ell_i$-th, otherwise $\rowspan{W'}$ would contain a row-vector that is not in $\rowspan{W}$. Conversely, $W'$ must contain a leading 1 at every $\ell_i$-th column, since $w_i\in\rowspan{W}$ for each $i$.
Thus, for each $0\le i<k$, $W'_{i,:}$ has a leading 1 at column $\ell_i$.

Finally, since $w_i\in\rowspan{W'}$ and, by definition of rref, $w_i$ contains a 1 at column $\ell_i$ and a 0 at all other leading columns $\ell_j$ for $j\ne i$, the only way for $w_i$ to be a linear combination of the rows of $W'$ is if $w_i=W'_{i,:}$. Thus, $W_{:k,:}=W'_{:k,:}$.
\end{proof}

\noindent The previous two lemmas yield an important result about \textit{full-rank factorizations}, i.e. rank-$r$ factorizations of a rank-$r$ matrix:

\begin{lemma}
\label{fullrank}
Consider an $m\times n$ matrix $M$ over an arbitrary field with rank $r:=\rank{M}>0$.
Then there exists $C_0\in\F^{m\times r},\ F_0\in\F^{r\times n}$ s.t. $C_0F_0=M$.
Furthermore, for any $C\in\F^{m\times r},\ F\in\F^{r\times n}$ s.t. $CF=M$, $\exists X\in\GL{r}{\F}$ s.t. $C=C_0 X$ and $F=X^{-1}F_0$.
\end{lemma}
\begin{proof}
By Lemma \ref{rref}, $\exists T\in\GL{m}{\F}$ s.t. $TM=\rref{M}$, so $T^{-1}\rref{M}=M$. Since all except the top $r$ rows of $\rref{M}$ are all-zeros, $(T^{-1})_{:,:r}\rref{M}_{:r,:}=M$, so we can set $(C_0,F_0)=((T^{-1})_{:,:r},\rref{M}_{:r,:})$.

Now suppose there are $C\in\F^{m\times r},\ F\in\F^{r\times n}$ s.t. $CF=M$.
Since $\rowspan{M}=\rowspan{CF}\subseteq\rowspan{F}$ and $\rank{F}\le r=\rank{M}$, we have $\rowspan{M}=\rowspan{F}$ and $\rank{F}=r$.
Since row reduction preserves rowspan, $\rowspan{\rref{M}}=\rowspan{\rref{F}}$, so by Lemma \ref{unique-rref}, $\rref{F}=\rref{M}_{:r,:}=F_0$.

Using Lemma \ref{rref}, $\exists X\in\GL{r}{\F}$ s.t. $XF=F_0$. Solving for $F$ and substituting it in the original equation $CF=M$ yields $(CX^{-1})F_0=M$. Since $\rank{F_0}=r$, its rows are linearly independent, so there is at most one solution for $CX^{-1}$. We know $C_0$ is one possible solution, so $CX^{-1}=C_0$. Thus, $C=C_0X$ and $F=X^{-1}F_0$.

Since invertible matrices form a group, it follows that all pairs of full-rank factorizations of $M$ are reachable from each other via the transformation $(C,F)\mapsto(CX,X^{-1}F)$.
\end{proof}

The upshot of this result is that there are exactly $|\GL{r}{\F}|\in O(|\F|^{r^2})$ many rank-$r$ decompositions of a rank-$r$ matrix, \textit{no matter how large $m$ and $n$ are}. 
This is the core reason why our main algorithm will run in polynomial time w.r.t. the size of the tensor.
In contrast, the number of rank-$(r+1)$ decompositions of a rank-$r$ matrix $M$ might be exponential w.r.t. the dimensions of $M$, e.g., when $M$ is all-zeros.

\noindent Finally, we deduce a property of matrix factorizations that are not necessarily full-rank:

\begin{lemma}
\label{lessrank}
For $U\in\F^{m\times r}$ and $V\in\F^{r\times n}$, let $r':=\rank{UV}$. Then $\rank{U}\le r'$ or $\rank{V}<r$.
\end{lemma}
\begin{proof}
By Lemma \ref{rref}, $\exists$ $T\in\GL{m}{\F}$ s.t. $TUV=\rref{UV}$. In particular, all but the topmost $r'$ rows of $TUV$ are all-zeros.

If $\rank{V}=r$, then the rows of $V$ are linearly independent.
Because $(TUV)_{i,:}=(TU)_{i,:}V=\vec{0} \ \forall i\ge r'$, this forces $(TU)_{i,:}=\vec{0} \ \forall i\ge r'$, so all except the top $r'$ rows of $TU$ are all-zeros. We then have $\rank{TU}\le r'$ and thus $\rank{U}\le r'$ due to $T$ being invertible.

We have thus proven that $(\rank{V}=r)\Rightarrow (\rank{U}\le r')$, which is logically equivalent to the final statement of the lemma.
\end{proof}

Using the well-known properties $\rank{A}=\rank{A^\intercal}$ and $\rank{AB}\le\min{\rank{A},\rank{B}}$ for arbitrary matrices $A$ and $B$, we can combine Lemma \ref{lessrank} with its transpose to get a stronger condition:

\begin{corollary}
\label{lessrank-both}
For $U\in\F^{m\times r}$ and $V\in\F^{r\times n}$, let $r':=\rank{UV}$.
Then $\paren{\rank{U}\le r' \ \lor \ \rank{V}<r}
\ \land \ \paren{\rank{U}<r \ \lor \ \rank{V}\le r'}
\ \land \ \rank{U}\ge r' \ \land \ \rank{V}\ge r'$.
\end{corollary}

\subsection{Algorithm}
We now present our main algorithm to solve $T_{i,j,k}=\sum_{r=0}^{R-1} A_{r,i}B_{r,j}C_{r,k} \ \forall i,j,k$, for an arbitrary $n\times n\times n$ tensor $T$, where $R\le 4$ is fixed.

Our main insight is to fix $A$ and then solve for $B,C$ in a clever manner. To ensure we do not enumerate too many $A$, we extract a basis on the slices $T_{i,:,:}$, denoted as $\{T'_{j,:,:}\}_j$, which can be obtained by flattening the slices $T_{i,:,:}$ and row-reducing the matrix $\~{T}:=\M{\vdots \\\hline \vec{T_{i,:,:}} \\\hline \vdots}_i$. Let $J:=\rank{\~{T}}$: then $T'$ contains $J$ slices, and by Lemma \ref{rref}, $\exists M\in\GL{n}{\F}$ s.t. $\sum_{0\le i<n} M_{j,i}T_{i,:,:} = \casew{T'_{j,:,:} & j<J \\ O}$, implying $T_{i,:,:}=\sum_{0\le j<J} (M^{-1})_{i,j}T'_{j,:,:}$.

These relations allow us to convert an arbitrary rank-$R$ decomposition of $T$ into a rank-$R$ decomposition of $T'$ and vice versa:

\begin{align*}
    T'=\sum_{r=0}^{R-1} A'_r\times B_{r,:}\times C_{r,:} & \Rightarrow T=\sum_{r=0}^{R-1} \bracket{\sum_{0\le j<J} (M^{-1})_{i,j}A'_{r,j}}_{0\le i<n} \times B_{r,:}\times C_{r,:} \\
    T=\sum_{r=0}^{R-1} A_r\times B_{r,:}\times C_{r,:} & \Rightarrow T'=\sum_{r=0}^{R-1} \bracket{\sum_{0\le i<n} M_{j,i}A_{r,i}}_{0\le j<J}\times B_{r,:}\times C_{r,:}
\end{align*}

Thus, our original problem is equivalent to solving $T'=\sum_r A'_{r,:}\times B_{r,:}\times C_{r,:}$.

The upshot of constructing $T'$ is that if the original problem $T=\sum_r A_r\times B_{r,:}\times C_{r,:}$ has a solution, then every $T_i$ is a linear combination of $\brace{B_{r,:}\times C_{r,:}}_{0\le r<R}$, so all $T_{i,:,:}$ collectively are in a linear subspace (of $\F^{n\times n}$) of rank at most $R$. This forces $T'$ to have at most $R$ many slices; otherwise, the original problem has no solution and we can immediately terminate. Thus, there are at most $|\F|^{R^2}$ possible $A'$ to consider, which is constant w.r.t. $n$.

If $J=0$, i.e. $T'$ is empty, then $T$ must be all-zeros, so it trivially has a rank-$R$ decomposition (e.g., with all elements in $A,B,C$ set to 0).

Otherwise, we proceed by treating each $B_{r,:}\times C_{r,:}$ as an unknown matrix $M_r$ with the condition that $\rank{M_r}\le 1$  \footnote{
    Reconstructing some $B_{r,:}$ and $C_{r,:}$ from a given $M_r$ s.t. $B_{r,:}\times C_{r,:}=M_r$ can be done via rank factorization.
}.

Then the decomposition problem on $T'$ becomes a linear system of equations over \textit{matrix} variables $M_r$ (instead of numbers): $\M{\vdots \\ T'_{j,:,:} \\ \vdots}_{0\le j<J}=A'^\intercal \M{\vdots \\ M_r \\ \vdots}_r$  \footnote{Note that we do not concatenate the matrices $T'_{j,:,:}$ or $M_r$ together; we organize them into column-vectors whose elements are themselves matrices.}.

As in any linear system, we can do row reduction: by Lemma \ref{rref}, $\exists F\in\GL{J}{\F}$ s.t. $FA'^\intercal=\rref{A'^\intercal}$.
Left-multiplying the linear system by $F$ yields $\M{\vdots\\D_j\\\vdots}_j=\rref{A'^\intercal}\M{\vdots \\ M_r \\ \vdots}_r$, where $D_j:=\sum_k F_{j,k}T'_{k,:,:} \ \in \F^{n\times n}$.
Because $F$ is invertible, we can left-multiply the new linear system by $F^{-1}$ to get the old system again, so any solution of the new system is also a solution of the old system and vice versa; thus, solving the new system over $M_r$ is equivalent to solving the old system.

Now we split casework over the matrix $\rref{A'^\intercal}$, which we abbreviate as $E$ from this point on. We pay special attention to where the nonzero elements of $E$ are located within each row, specifically by analyzing the sets $S_i:=\brace{r:E_{i,r}\ne 0}$, i.e. the sets of column indices of nonzero elements within each row. Using these sets, our system of equations is equivalent to $D_i=\sum_{r\in S_i} E_{i,r} M_r \ \forall i$, under the constraints $\rank{M_r}\le 1 \ \forall r$.

\subsubsection{Preprocessing}
First, we can remove any rows in $E$ that are all-zeros, since if there is such a row $j$, then having $D_j\ne O$ will make the system have no solution, and having $D_j=O$ will make this row redundant.

Next, order for a solution to exist, we must have $\rank{D_i}\le|S_i| \ \forall i$, since each $D_i$ is the sum of $|S_i|$ many matrices that each have rank $\le 1$. We call this the \textit{rank precondition}.

Finally, we notice that if there is some set of indices $I$ that does not share any $M_r$ variables with any outside rows, i.e. $\forall i\in I \ \forall j\not\in I : S_i\cap S_j = \emptyset$, then solving rows $i\in I$ of the system over variables $\brace{M_r}_{r\in \cup_{i\in I} S_i}$ is independent of solving rows $i\not\in I$ of the system over variables $\brace{M_r}_{r\in \cup_{i\not\in I} S_i}$.
We say $E$ is \textit{separable} if this scenario is possible  \footnote{
    An example of $E$ being separable is $E=\M{1&&&x\\&1\\&&1&y}$, since row 1 can be separated from rows $\brace{0,2}$.
}.

Thus, we can assume WLOG that $E$ is a rref matrix that contains $R$ columns and is not separable  \footnote{
    Algorithmically, we can separate $E$ into non-separable pieces by finding the connected components of the graph $G=\paren{ \brace{0,\dots J-1},\ \brace{(i,j):S_i\cap S_j\ne\emptyset} }$, which can be done in $O(R)$ time using breadth-first search.
}.
This forces $E$ to be restricted to the following forms
\footnote{
    One could further restrict $E$ by permuting the variables $M_r$ and rescaling them by nonzero numbers, which corresponds to right-multiplying $E$ by $PD$, where $P$ is a permutation matrix and $D$ is a diagonal matrix s.t. all elements on the diagonal are nonzero.
}, where $w,x,y,z$ below are arbitrary elements of $\F$. Blank elements denote zeros.
\begin{itemize}
    \item $R=1$:
    $\M{1}$

    \item $R=2$:
    $\M{1&x}$

    \item $R=3$:
    $\M{1&x&y},\ \M{1&&x\\&1&y}$

    \item $R=4$:
    $\M{1&x&y&z},\ \M{1&&w&x\\&1&y&z},\ \M{1&x&&y\\&&1&z},\ \M{&1&&x\\&&1&y},\ \M{1&&&x\\&1&&y\\&&1&z}$
\end{itemize}

\noindent We group these forms into three types:

\subsubsection{Single row}
In this case, we only need to solve for $D_0$, so there exists a solution iff $\rank{D_0}\le|S_0|$ (i.e. the rank precondition is satisfied), and such a solution can be constructed via rank factorization.

\subsubsection{Rows sharing one common column}
In this case, each pair of $S_i$ intersects at exactly one index, and that index is the same across all pairs, i.e. there is some $r^*$ s.t. for all pairs of row indices $i,i'$ with $i\ne i'$, $S_i\cap S_{i'} = \brace{r^*}$. This case covers $E=\M{1&&x\\&1&y},\ \M{1&x&&y\\&&1&z},\ \M{&1&&x\\&&1&y},\ \M{1&&&x\\&1&&y\\&&1&z}$, where all elements in the rightmost column are nonzero.

Here, we utilize full-rank matrix factorizations. If $\exists i \textrm{ s.t. } \rank{D_i}=|S_i|$, we can enumerate all rank-$|S_i|$ factorizations of $D_i$, then for each factorization, check if $\rank{D_{i'}-E_{i',r^*}M_{r^*}}\le |S_{i'}-1| \ \forall i'\ne i$; if this is satisfied, a solution exists, which can again be constructed with repeated matrix rank factorization.

If there is no such $i$, then $\rank{D_i}\le |S_i|-1 \ \forall i$, so we can force $M_{r^*}=O$ and solve each $D_i$ independently.



\subsubsection{Rows sharing two common columns}
This case covers the last remaining form, $E=\M{1&&w&x\\&1&y&z}$, and is the most difficult. Written in full, our system of equations is $\brace{\begin{array}{cccc}
    D_0 = & M_0 && +wM_2+xM_3 \\
    D_1 = && M_1 & +yM_2+zM_3
\end{array}}$, under the constraints $\rank{M_r}\le 1 \ \forall r$.

\subsubsection*{Fixing coefficients WLOG}

WLOG assume $w,x,y,z\ne 0$, otherwise we revert to one of the previous cases (after separating $E$ into independent parts).

We can substitute $M_2':=wM_2,\ M_3':=xM_3,\ y':=\frac{y}{w},\ z':=\frac{z}{x}$ to get $\brace{\begin{array}{cccc}
    D_0 = & M_0 && +M_2'+M_3' \\
    D_1 = && M_1 & +y'M_2'+z'M_3'
\end{array}}$,
then substitute $D_1':=\frac{1}{y'}D_1,\ M_1':=\frac{1}{y
}M_1,\ z'':=\frac{z'}{y'}$ to get $\brace{\begin{array}{cccc}
    D_0 = & M_0 && +M_2'+M_3' \\
    D_1' = && M_1' & +M_2'+z''M_3'
\end{array}}$.
Since $\rank{M}=\rank{\sigma M}$ for any $\sigma\ne 0$, solving our new system of equations with the constraint that each of $M_0,M_1',M_2',M_3'$ has rank $\le 1$ is equivalent to solving our original system with the constraint that each $M_r$ has rank $\le 1$.

Thus, WLOG we can assume $w,x,y=1$ in our original system.

\subsubsection*{If $\boldsymbol{z=1}$:}
By subtracting the second equation from the first, our system becomes $\brace{\begin{array}{cccc}
    D_*= & M_0 & -M_1 & \\
    D_1= && M_1  & +M_2+M_3
\end{array}}$, where $D_*:=D_0-D_1$, so it can be solved with the rows-sharing-one-common-column case.

\subsubsection*{Otherwise, $\boldsymbol{z\ne 0,1}$:}

\noindent We can solve several easy cases on $\rank{D_0}$ and $\rank{D_1}$.
WLOG assume the rank precondition: $\rank{D_0}\le 3 \ \land \ \rank{D_1}\le 3$.

\begin{itemize}
    \item If $\rank{D_0}=3$:

    Enumerate each of the $O(|\F|^9)$ possible full-rank factorizations of $D_0$, which fixes $(M_0,M_2,M_3)$;
    
    then check if $\rank{D_1-M_2-zM_3}\le 1$.

    \item Else if $\rank{D_1}=3$:

    Likewise, enumerate each full-rank factorization of $D_1$, which fixes $(M_1,M_2,M_3)$;
    
    then check if $\rank{D_0-M_2-zM_3}\le 1$.

    \item Else: we have $\rank{D_0}\le 2$ and $\rank{D_1}\le 2$:
    
    \begin{itemize}
        \item If $\rank{D_0}\le 1$:

        Set $M_0=D_0$;

        construct an arbitrary rank-2 factorization $D_1=x_0y_0+x_1y_1$, where $x_i$ are column-vectors and $y_i$ are row-vectors;
        then set $M_1=x_0y_0$.

        We now must solve $M_2+M_3=O,\ M_2+zM_3=x_1y_1$.
        This yields $M_2=-M_3$, and since $z\ne 1$, $M_3=\frac{1}{z-1}x_1y_1$, so we always have a solution:

        \[\paren{M_0,M_1,M_2,M_3}=\paren{
            D_0,
            x_0y_0,
            -\frac{1}{z-1}x_1y_1,
            \frac{1}{z-1}x_1y_1
        }.\]

        \item Else if $\rank{D_1}\le 1$:

        Likewise, set $M_1=D_1$, factorize $D_0=x_0y_0+x_1y_1$,  set $M_0=x_0y_0$. Then we must solve $M_2+M_3=x_1y_1,\ M_2+zM_3=O$, which can be done by setting $M_3=-\frac{1}{z-1}x_1y_1,\ M_2=-zM_3$.
    \end{itemize}
\end{itemize}

This leaves $\rank{D_0}=2 \ \land \ \rank{D_1}=2$ as the only remaining case. We can substitute $M_i=u_iv_i$, with $u_i\in\F^{n\times 1},\ v_i\in\F^{1\times n}$, since $\rank{M_i}\le 1$.
Then solving the system is equivalent to solving $D_0=\underbrace{\MA{c|c|c}{u_0&u_2&u_3}}_{=:U_0}\underbrace{\MA{c}{v_0\\\hline v_2\\\hline v_3}}_{=:V_0}$ and $D_1=\underbrace{\MA{c|c|c}{u_1&u_2&u_3}}_{=:U_1}\underbrace{\MA{c}{v_1\\\hline v_2\\\hline zv_3}}_{=:V_1}$ simultaneously.
Invoking Corollary \ref{lessrank-both}, we have:
\begin{itemize}
    \item $\rank{D_0}=2 \Rightarrow (\rank{U_0},\rank{V_0})\in\brace{(3,2),(2,3),(2,2)}$
    \item $\rank{D_1}=2 \Rightarrow (\rank{U_1},\rank{V_1})\in\brace{(3,2),(2,3),(2,2)}$
\end{itemize}

Collectively, we have $\paren{\rank{U_0}=2 \ \lor \ \rank{V_0}=2} \ \land \ \paren{\rank{U_1}=2 \ \lor \ \rank{V_1}=2}$.

Although the factorizations $D_i=U_iV_i$ are not full-rank, at least one of $U_i$ and $V_i$ must have the same rank as $D_i$, which motivates the following lemma:

\begin{lemma}
\label{not-fullrank}
For $U\in\F^{m\times r},\ V\in\F^{r\times n}$, define $D:=UV$ and $r':=\rank{D}$.
Let $D=U_*V_*$ be a full-rank factorization of $D$, with $U_*\in\F^{m\times r'},\ V_*\in\F^{r'\times n}$.
Suppose furthermore that $\rank{U}=r'$.
Then $U=U_*G$ for some $G\in\F^{r\times r'}$.
\end{lemma}
\begin{proof}
By Lemma \ref{fullrank}, there exists a full-rank factorization of $U=U'F$, where $U'\in\F^{m\times r'},\ F\in\F^{r'\times r}$.
Then $D=U'(FV)$ is a full-rank factorization of $D$, so by Lemma \ref{fullrank} again, $\exists X\in\GL{r'}{\F}$ s.t. $(U',FV)=(U_*X,X^{-1}V_*)$, which yields $U=U_*(XF)$. \footnote{
    Notice that $V$ is much less restricted, as it only has to satisfy is $(XF)V=V_*$, which could have exponentially many solutions w.r.t. $m,n$.
}
\end{proof}

Thus, we have if $UV$ is a rank-$r$ factorization of a rank-$r'$ matrix, and that $\rank{U}=r$, then there are only $O(|\F|^{rr'})$ possibilities for $U$, and each can be enumerated in polynomial time via matrix multiplication. By matrix transpose, a similar result applies to $V$ if $\rank{V}=r$. Once again, we avoid exponential dependence on the dimensions of the matrices.

We can now solve our last case by splitting it into the following four subcases:
\begin{itemize}
    \item If $\rank{U_0}=2 \ \land \ \rank{U_1}=2$:

    Enumerate all possible pairs $(U_0,U_1)$, which fixes the column-vectors $u_i$;
    
    for each pair, if it does not cause conflicts on what each $u_i$ gets fixed to (i.e. if $(U_1)_{:,2}=(U_0)_{:,2}$ and $(U_1)_{:,3}=(U_0)_{:,3}$), solve the linear system $\MA{c}{D_0\\\hline D_1}=\MA{c|c|c|c}{u_0&\vec{0}&u_2&u_3\\\hline \vec{0}&u_1&u_2&zu_3}\MA{c}{v_0\\\hline v_1\\\hline v_2\\\hline v_3}$ for $v_i$.

    \item Else if $\rank{V_0}=2 \ \land \ \rank{V_1}=2$:

    Enumerate all possible $(V_0,V_1)$, which fixes the row-vectors $v_i$;
    \begin{itemize}
        \item Notice that this step relies on $z$ being nonzero, as otherwise it is possible for $v_3$ to still have exponentially many possibilities.
    \end{itemize}

    for each pair, if it does not cause conflicts on $v_i$, solve $\MA{c|c}{D_0&D_1}=\MA{c|c|c|c}{u_0&u_1&u_2&u_3}\MA{c|c}{v_0&\vec{0}\\\hline \vec{0}&v_1\\\hline v_2&v_2\\\hline v_3&zv_3}$ for $u_i$.

    \item Else if $\rank{U_0}=2 \ \land \rank{V_1}=2$:

    Enumerate all possible $(U_0,V_1)$, which fixes all variables except $u_1$ and $v_0$;

    for each pair, solve $u_0v_0=D_0-u_2v_2-u_3v_3$ for $v_0$ and $u_1v_1=D_1-u_2v_2-zu_3v_3$ for $u_1$.

    \item Else: must have $\rank{U_1}=2 \ \land \rank{V_0}=2$:

    Enumerate all possible $(U_1,V_0)$;

    for each pair, solve $u_0v_0=D_0-u_2v_2-u_3v_3$ for $u_0$ and $u_1v_1=D_1-u_2v_2-zu_3v_3$ for $v_1$.
\end{itemize}

In each subcase, when enumerating matrix pairs, we apply Lemma \ref{not-fullrank} with $r=3$ and $r'=2$, so each matrix in the pair has $O(|\F|^6)$ many possibilities, resulting in the pair having $O(|\F|^{12})$ many possibilities.

\subsection{Runtime}
Since $R\le 4$, we will omit $R$ terms in some places to simplify analysis.

To provide a good upper bound on the asymptotic time complexity of our algorithm, we take advantage of the fact that all matrices we deal with must have low rank in order for a solution to the original tensor decomposition problem to exist. For this, we utilize an efficient algorithm that returns full-rank factorizations for matrices of rank below some threshold and returns nothing otherwise.

\begin{lemma}
\label{fast-rankfac}
Given a matrix $M\in\F^{m\times n}$ and integer $r'$, we can do the following in $O(m^2+m(m+n)\min{r',\rank{M}})$ time:
\begin{itemize}
    \item If $\rank{M}\le r'$: return some $C\in\F^{m\times \rank{M}}$ and $F\in\F^{\rank{M}\times n}$ s.t. $M=CF$
    \item Else: return nothing
\end{itemize}
\end{lemma}
\begin{proof}
We run Gaussian elimination with early termination:
\begin{enumerate}
    \item initialize $(C,F)\gets (I_m,M)$ and $r\gets 0$:
    \item for $j$ from $0$ to $n-1$, inclusive:
    \begin{enumerate}
        \item if there is no $r\le i<m$ s.t. $F_{i,j}\ne 0$: continue to the next iteration of the loop
        \item else:
        \begin{enumerate}
            \item find some $r\le i<m$ s.t. $F_{i,j}\ne 0$
            
            \item swap rows $i$ and $r$ of $F$
            \item swap columns $i$ and $r$ of $C$

            \item let $\sigma=F_{r,j}$
            \item set $F_{r,:} \gets F_{r,:}/\sigma$
            \item set $C_{:,r} \gets \sigma C_{:,r}$
            
            \item for $k$ from $r+1$ to $m-1$, inclusive:
            \begin{enumerate}
                \item let $s=F_{k,j}$
                \item set $F_{k,:}\gets F_{k,:}-s F_{r,:}$
                \item set $C_{:,r}\gets C_{:,r}+s C_{:,k}$
            \end{enumerate}
            \item increment $r$
            \item if $r>r'$: terminate and return nothing
        \end{enumerate}
    \end{enumerate}
    \item return $(C_{:,:r},F_{:r,:})$
\end{enumerate}

If there was no early termination (ix.), this algorithm would transform $F$ into row echelon form (not necessarily \textit{reduced} row echelon form), so it would reach step (b) exactly $\rank{M}$ many times. Thus, with early termination, the algorithm will reach step (b) $\min{r',\rank{M}}$ many times,  correctly terminating if $\rank{M}>r'$.

Finally, each execution of step (b) involves $O(m)$ many row operations in $F$ and column operations in $C$, each of which takes $O(m+n)$ time.
Since initializing $T$ takes $O(m^2)$ time, the total runtime is $O(m^2+m(m+n)\min{r',\rank{M}})$.
\end{proof}

Throughout our algorithm for tensor decomposition, we maintain the invariants that $T$ is invertible and $TM=F$. At the end of the algorithm, $F=\rref{M}$.

In our algorithm, we first extract a basis of tensor slices $T_{i,:,:}$ by row-reducing the $n\times n^2$ matrix $\~{T}:=\MA{c}{\vdots \\\hline \vec{T_{i,:,:}} \\\hline \vdots}_i$. We also check whether $\rank{\~{T}}\le R$ and immediately exit if this condition is not satisfied. Using Lemma \ref{fast-rankfac}, we can complete this entire step in $O(Rn^3)$ time, instead of the more na\"ive $O(n^4)$ time.

Then, for each possible $A'$, we solve a system of linear equations over $R$ many matrices $M_r$ of rank $\le 1$.
Since $A'$ has $R$ columns and $\le R$ rows, there are $O(|\F|^{R^2})$ many of them.

To solve the system over $M_r$, we first separate it into independent parts, of which there are $\le R$ many.
To solve each part, we consider the most difficult subcases for each $R$:
\begin{itemize}
    \item If $R\le 2$: all coefficient matrices $E$ are separable into single rows, and we must have $\rank{D_i}\le R$ in order for a solution to exist, so we can find a solution (or prove it does not exist) in $O(n^2)$ time.
    \item If $R=3$: the most difficult case occurs when $E=\M{1&&x\\&1&y}$ and at least one of $D_0$ or $D_1$ has rank 2, where we have to enumerate $O(|\F|^4)$ many full-rank factorizations of such a $D_i$.
    For each factorization of $D_0$ (or $D_1$), we check if $D_0-xM_2$ (or $D_1-yM_2$) has rank $\le 1$. Overall, solving this case takes $O(|\F|^4 n^2)$ time.
    \item If $R=4$: the most difficult case occurs when $E=\M{1&&w&x\\&1&y&z}$ and both $D_0$ and $D_1$ have rank 2, where we have to enumerate $O(|\F|^{12})$ many pairs of matrices in 4 subcases each.

    Generating each pair of matrices takes $O(n)$ time, and for each pair we must solve for $X$ in some linear system $AX=B$ (or its transpose), where $A$ has shape $O(n)\times O(R)$ and $X$ has shape $O(R)\times O(n)$.

    We can find $X$ with row reduction: find some invertible $T$ s.t. $TA=\rref{A}$ and extract a solution $X$ from the row-reduced augmented matrix $T\MA{c|c}{A&B}=\MA{c|c}{\rref{A}&TB}$, e.g., by checking for all-zeros rows in $\rref{A}$ and setting free variables of $X$ to 0.

    We can obtain such $T$ in $O(n^2R)$ time by modifying the algorithm in Lemma \ref{fast-rankfac}: remove early termination, maintain the inverse of the $C$ matrix by performing row operations on a separate matrix, and continue the row reduction so that the $F$ matrix is transformed into rref.

    Overall, solving this subcase takes $O(|\F|^{12} n^2)$ time.
\end{itemize}

Thus, we spend $O(|\F|^{g(R)}n^2)$ time solving each system over $M_r$, where $g(1)=0,\ g(2)=0,\ g(3)=4,\ g(4)=12$.

We can then upper bound the total runtime of our algorithm at $O\paren{Rn^3 + |\F|^{R^2}R\cdot |\F|^{g(R)}n^2}$, which for fixed $R$ is equivalent to $O\paren{n^3 + |\F|^{R^2+g(R)}n^2}$. Notice that the large constant is multiplied with $n^2$, not $n^3$. It may be possible to obtain tighter bounds through more careful analysis.

\section{With wildcards}
When certain elements in $T$ are wildcards (i.e., allowed to have arbitrary value), tensor decomposition becomes much more difficult. In this section we prove Theorem \ref{NAE}, i.e., rank-2 decomposition with wildcards over $\ZM{2}$ is NP-hard, via a reduction from NAE-3SAT.

Let $\mathcal{I}:=\M{\M{1&0\\0&0}&\M{0&0\\0&1}}$ be the $2\times 2\times 2$ tensor with 1s along the main diagonal and 0s everywhere else. It can be verified by brute force that the only rank-2 decompositions of $\mathcal{I}$ are $\M{1\\0}^{\times 3}+\M{0\\1}^{\times 3}$ and $\M{0\\1}^{\times 3}+\M{1\\0}^{\times 3}$ (even when allowing all-zero vectors). We can parameterize both of these solutions as $\M{a\\\neg a}^{\times 3}+\M{\neg a\\a}^{\times 3}$, yielding a variable gadget. By concatenating $n$ copies of $\mathcal{I}$ diagonally \footnote{For an example, the result of concatenating
$\M{\M{a_{0}&a_{1}\\a_{2}&a_{3}}&\M{a_{4}&a_{5}\\a_{6}&a_{7}}}$ and
$\M{\M{b_{0}&b_{1}\\b_{2}&b_{3}}&\M{b_{4}&b_{5}\\b_{6}&b_{7}}}$ diagonally is

\[
\M{
\M{a_{0}&a_{1}\\a_{2}&a_{3}\\&&&\\&&&}
&\M{a_{4}&a_{5}\\a_{6}&a_{7}\\&&&\\&&&}
&\M{&&&\\&&&\\&&&\\&&b_{0}&b_{1}\\&&b_{2}&b_{3}}
&\M{&&&\\&&&\\&&&\\&&b_{4}&b_{5}\\&&b_{6}&b_{7}}
},
\]

where the unmarked cells can be filled with arbitrary things.} to form the target tensor $T$ and setting all other elements to be wildcards, we force all rank-2 decompositions of $T$ to be of the form $\M{\vdots\\v_i\\\neg v_i\\\vdots}^{\times 3}+\M{\vdots\\\neg v_i\\v_i\\\vdots}^{\times 3}$, for arbitrary variables $v_i$.





Since $\NAE{a,b,c}=\neg (abc \lor (\neg a)(\neg b)(\neg c))$, and the AND-clauses $abc$ and $(\neg a)(\neg b)(\neg c)$ cannot both be equal at the same time, $\NAE{a,b,c}=\neg(abc \oplus (\neg a)(\neg b)(\neg c))
=(abc+(\neg a)(\neg b)(\neg c)=0)$, where $\oplus$ denotes XOR and is isomorphic to addition mod 2.
Because $T$ already forces all elements of vectors in the second outer product to be negations of vectors in the first outer product, we can encode $\NAE{v_i,v_j,v_k}$ by setting $T_{2i,2j,2k}=0$; for negative literals, we add 1 to the corresponding tensor index, e.g., $\NAE{v_i,v_j,\neg v_k}$ can be encoded by setting $T_{2i,2j,2k+1}=0$.

Thus, we convert a NAE-3SAT instance with $n$ variables and $m$ clauses into a rank-2 decomposition instance on a $2n\times 2n\times 2n$ tensor with $8n+m$ fixed elements and all other elements being wildcards.

Over arbitrary finite fields, there is an issue of scale invariance: each outer product $a\times b\times c$ can be replaced with $(\alpha a)\times (\beta b)\times (\gamma c)$ for any $\alpha,\beta,\gamma$ s.t. $\alpha\beta\gamma=1$, which allows a single tensor decomposition to be transformed into many other tensor decompositions that evaluate to the same tensor. This scaling freedom makes constructing gadgets for NP-hardness more difficult.

\subsection{Rank 1 with wildcards}
For rank-\textit{1} decomposition with wildcards, we give a polynomial-time algorithm over a fixed finite field $\F$.

First, we rewrite the problem as follows: given an arbitrary subset $S\subseteq \brace{0,\dots n-1}^3$ of tensor indices and a set of elements $T_{i,j,k}\in\F \ \forall (i,j,k)\in S$, find variables $a_{0\dots n-1}, b_{0\dots n-1}, c_{0\dots n-1}$ s.t. $a_i b_j c_k=T_{i,j,k} \ \forall (i,j,k)\in S$.

When $\F=\Z/2\Z$, $abc=1$ implies $(a,b,c)=(1,1,1)$, so a solution exists if and only if setting
\[a_i=\lor_{j,k:(i,j,k)\in S} T_{i,j,k},\ 
b_j=\lor_{i,k:(i,j,k)\in S} T_{i,j,k},\ 
c_k=\lor_{i,j:(i,j,k)\in S} T_{i,j,k}\] yields a solution; this is because this construction produces the minimal set of elements in $a\times b\times c$ that are forced to be 1s by the values $T_{i,j,k}$, and this set covers all $(i,j,k)\in S$ that have $T_{i,j,k}=1$.
Thus for $\F=\Z/2\Z$ there is an algorithm for rank-1 wildcard decomposition that runs in $O(n^3)$ time with low constant factor.


For an arbitrary fixed finite field $\F$, we outline a polynomial-time algorithm.

\noindent Let $S':=\brace{(i,j,k)\in S:T_{i,j,k}\ne 0}$ denote all points with nonzero values and $I:=\brace{i:\exists (i,j,k)\in S'}$, $J:=\brace{j:\exists (i,j,k)\in S'}$, $K:=\brace{k:\exists (i,j,k)\in S'}$ denote the sets of coordinate values of the nonzero points along each dimension. Since $T_{i,j,k}\ne 0 \Rightarrow a_i,b_j,c_k\ne 0$, we have that $a_i\ne 0 \ \forall i\in I$, $b_j\ne 0 \ \forall j\in J$, and $c_k\ne 0 \ \forall k\in K$, so we check whether $T_{i,j,k}\ne 0 \ \forall (i,j,k)\in (I\times J\times K)\cap S$.
If this check fails, there is no solution.

Otherwise, we observe that all the equations $a_i b_j c_k=T_{i,j,k}$ involve multiplications only. This suggests we replace all scalars with exponentials under the same base, as this turns multiplications into additions.

Indeed, it is well-known that the multiplicative group of any finite field is cyclic, so there exists some $\mu\in\F$ s.t. its powers cover all of $\F\setminus\brace{0}$. We can then substitute the variables $a_i=\mu^{p_i},\ b_j=\mu^{q_j},\ c_k=\mu^{r_k}$ and tensor elements $T_{i,j,k}=\mu^{t_{i,j,k}}$ to get the equations $p_i+q_j+r_k \equiv t_{i,j,k} \ (\bmod |\F|-1) \ \forall (i,j,k)\in S'$. Thus, solving our problem is equivalent to solving a linear system over $\ZM{(|\F|-1)}$.

From \cite{lin-mod}, solving the matrix-vector equation $xA=b$ over $\ZM{N}$ for $n\times m$ $A$ can be done in $O(nm^{\omega-1})$ many operations, where $\omega\le 2.371552$ \cite{mm-record} is the exponent for matrix multiplication, with all integers throughout the algorithm having magnitude $\le N$. 
Since our system consists of $|S'|$ equations over $3n$ variables, it can be solved in $O(n |S'|^{\omega-1})$ operations, which is within $O(n^{3\omega -2})$ since $|S'|\le n^3$.

The main bottleneck in this algorithm is finding a primitive element of $\F$, for which an efficient algorithm is not known \cite{primitive-elem}.
On the other hand, for $\F=\F_{p^k}$, the total input size of a problem instance would be linear w.r.t. $\log(p)$ and $k$ \footnote{Every element in $\F_{p^k}$ is a linear combination of $1,\alpha,\dots,\alpha^{k-1}$ with coefficients in $\ZM{p}$, for some element $\alpha$, so every element can be specified with a $k$-element vector over $\ZM{p}$.}. Because of this, we are not sure what the computational hardness of rank-1 wildcard decomposition is when $p,k$ are not fixed.


\subsection{Rank-1 with wildcards for matrices}
One may wonder whether \textit{matrix} rank factorization with wildcards is NP-hard. The issue is that matrix factorizations have an immense degree of freedom, as $CF=M$ implies $(CX^{-1})(XF)=M$, so constructing gadgets seems very difficult even for rank $R=2$.

Additionally, for rank $R=1$, we sketch an algorithm that runs in $O(n^3)$ time even when $\F$ is not fixed (assuming each operation between elements of $\F$ runs in $O(1)$ time), in contrast to the analogous problem for 3D tensors.
Again we formulate the problem as follows: given a subset $S\subseteq\brace{0,\dots n-1}^2$ and set of elements $T_{i,j}\in\F \ \forall (i,j)\in S$, find variables $a_{0,\dots n-1}, b_{0,\dots n-1}$ s.t. $a_i b_j=T_{i,j}$.

Like before, define $S':=\brace{(i,j)\in S: T_{i,j}\ne 0}$, $I:=\brace{i:(i,j)\in S'}$, $J:=\brace{j:(i,j)\in S'}$. Then $a_i\ne 0 \ \forall i\in I$ and $b_j\ne 0 \ \forall j\in J$, so we check whether $T_{i,j}\ne 0 \ \forall (i,j)\in (I\times J)\cap S$; if the check fails then there is no solution.

Otherwise, we notice that since every equation $a_i b_j=T_{i,j}$ constrains the product of two variables, knowing one variable immediately forces the value of the other variable.
This motivates us to construct a graph $G$ with vertex set $S'$ where there is an edge between every pair of vertices $(i,j),(i',j')$ s.t. $i=i'$ or $j=j'$.

Choose an arbitrary $(i_0,j_0)\in S'$: WLOG we can set $a_{i_0}=1, b_{j_0}=T_{i_0,j_0}$ because any solution $(a,b)$ to the original problem can be scaled to $(\alpha a,\frac{1}{\alpha} b)$ for any $\alpha\ne 0$. Then do a breadth-first search over $S'$ starting from $(i_0,j_0)$: when we visit each new vertex $(i,j)\in S$, we will have already determined $a_i$ and $b_j$, so for each neighbor of the form $(i',j)$ we set $a_{i'}=T_{i',j}/b_j$, and for each neighbor of the form $(i,j')$ we set $b_{j'}=T_{i,j'}/a_i$. Each neighbor that was not already visited is added to a queue of vertices to visit.
Note that because every $T_{i,j}$ we deal with is nonzero, we will never encounter division by 0 in this procedure.

If at any point we update a variable $a_i$ or $b_j$ that already has a assigned value and set it to a different value, then we have encountered an inconsistency and there is no solution to the original problem. Otherwise, at the end of this procedure we will have searched through the entire connected component of $G$ that contains $(i_0,j_0)$, so we do another breadth-first search starting at an arbitrary unvisited vertex in $G$, and repeat until all vertices in $G$ have been visited. If no inconsistencies were ever encountered, we can set all variables $a_i$ and $b_j$ that were not already assigned a value to 0, and we have a solution.

Since breadth-first search over a graph with $V$ vertices and $E$ edges takes $O(V+E)$ time, and $G$ has $O(n^2)$ vertices and $O(n^3)$ edges, our algorithm runs in $O(n^3)$ time.

The key for correctness is that if two vertices $(i,j),(i',j')$ are in different connect components, then $(i',j')$ is not adjacent to any vertex $(x,y)$ in the connected component containing $(i,j)$, so $i'\ne x$ and $j'\ne y$. This means that $a_{i'}$ and $b_{j'}$ are not directly or indirectly involved in any equations with elements $T_{x,y}$ for each $(x,y)$ in the connected component of $(i,j)$. Thus, after we finish searching over each connected component, all remaining unassigned variables $a_i$ and $b_j$ are completely independent of all previously assigned variables, so as we search through the next connected component we will never encounter inconsistencies with variables assigned during a previous connected component. This independence between connected components also allows us to assume WLOG that $a_{i_0}=1$ and $b_{j_0}=T_{i_0,j_0}$ at each individual connected component.

If we tried generalizing this algorithm to 3D tensors, we would encounter a problem: two vertices would have to match 2 out of 3 coordinates in order to be adjacent to each other, because each equation $a_i b_j c_k=T_{i,j,k}$ constrains the product of three variables instead of two, so two variables have to be known beforehand to force the value of the third one.
But then if $(i,j,k)$ and $(i',j',k')$ occupy different components, it is still possible for $(i',j',k')$ to share 1 out of 3 coordinates with a vertex contained in the connected component containing $(i,j,k)$, so different connected components can still share common variables and thus interact with each other in a nontrivial manner.




\section{Future directions}
We conclude with some open questions:
\begin{itemize}
    \item For what $R\ge 5$ can rank-$R$ tensor decomposition (without wildcards) be computed in polynomial time over a fixed finite field?
    \item For $R\le 4$, how low can the constant factor of rank-$R$ decomposition be w.r.t. to the size of the finite field and w.r.t. $R$?
    \item How asymptotically low can the number of wildcards in a tensor be while keeping rank-2 wildcard decomposition NP-hard?
    \item What is the parameterized complexity of rank-1 and rank-2 wildcard decomposition w.r.t. the number of wildcard elements in the target tensor?
    \item Can we have efficient tensor decomposition over number systems beyond finite fields, such as the integers?
\end{itemize}

\section*{Acknowledgments}
We thank Prof. Erik Demaine and TAs Josh Brunner, Lily Chung, and Jenny Diomidova for running MIT 6.5440 Fall 2023, which made this project possible. We also thank them for giving feedback on an earlier version of this report, some of which were requests for certain details and open questions to be addressed, including rank-1 wildcard decomposition.

\end{document}